\documentclass{sig-alternate-05-2015}

\newfont{\mycrnotice}{ptmr8t at 7pt}
\newfont{\myconfname}{ptmri8t at 7pt}
\let\crnotice\mycrnotice%
\let\confname\myconfname%

\usepackage[pass]{geometry}
\usepackage[utf8]{inputenc}   
\usepackage[T1]{fontenc}      
\usepackage{lmodern}
\mathcode`l="8000
\begingroup
\makeatletter
\lccode`\~=`\l
\DeclareMathSymbol{\lsb@l}{\mathalpha}{letters}{`l}
\lowercase{\gdef~{\ifnum\the\mathgroup=\m@ne \ell \else \lsb@l \fi}}%
\endgroup

\usepackage{amssymb,amsmath}
\usepackage{amsfonts}
\usepackage{color}
\usepackage{url}
\usepackage[colorlinks=true,citecolor=blue,hypertexnames=false]{hyperref}
\usepackage{enumerate}
\usepackage{comment}
\usepackage{mathrsfs}
\usepackage{bm}
\usepackage{stmaryrd}
\usepackage{array,multirow}

\usepackage[shortlabels]{enumitem}
\setlist[description]{font=\normalfont\itshape,itemsep=0ex,partopsep=0ex}

\usepackage{microtype}  
\usepackage{mathtools}  
\usepackage{booktabs}   
\usepackage{array}      
\usepackage[nocompress,nosort]{cite}
\usepackage{float}

\DeclareBoldMathCommand{\bbeta}{\beta}

\makeatletter
\def\thebibliography#1{%
\ifnum\addauflag=0\addauthorsection\global\addauflag=1\fi
     \section[References]{
        {References} 
          {\vskip -5pt plus 1pt} 
         \@mkboth{{\refname}}{{\refname}}%
     }%
     \tiny\list{[\arabic{enumi}]}{%
         \settowidth\labelwidth{[#1]}%
         \leftmargin\labelwidth
         \advance\leftmargin\labelsep
         \advance\leftmargin\bibindent
         \parsep=0pt\itemsep=1pt 
         \itemindent -\bibindent
         \listparindent \itemindent
         \usecounter{enumi}
     }%
     \let\newblock\@empty
     \raggedright 
     \sloppy
     \sfcode`\.=1000\relax
}
\makeatother

\definecolor{bleu}{rgb}{0,0,0.8}
\definecolor{orange}{rgb}{0.8,0.4,0}

\def\gathen#1{{#1}}

\newtheorem{theo}{Theorem}
\newtheorem{lem}[theo]{Lemma}
\newtheorem{prop}[theo]{Proposition}

\newcommand{\N}{\mathbb N}
\newcommand{\F}{\mathbb F}

\newcommand{\Q}{\mathbb Q}

\newcommand{\softO}{O\tilde{~}}

\newcommand{\Sdp}{\ell[[t]]^{\mathrm{dp}}}

\title{Computation of the Similarity Class of the p-Curvature\titlenote{\small %
We warmly thank the referees
for their very helpful comments, and for pointing out a mistake in an earlier version of the article.
We thank 
M. Giesbrecht, G.~Labahn and 
A. Storjohann for useful discussions. The third author is supported by NSERC.
\vspace{-25pt}}}

\numberofauthors{3} 
\author{
\alignauthor Alin Bostan\\
\affaddr{Inria (France)}\\
  \affaddr{\textsf{alin.bostan@inria.fr}}
\alignauthor Xavier Caruso\\
  \affaddr{Universit\'e Rennes 1 (France)}\\
  \affaddr{\textsf{xavier.caruso@normalesup.org}}
\alignauthor \'Eric Schost\\
  \affaddr{Univ. of Waterloo (Canada)}\\
  \affaddr{\textsf{eschost@uwaterloo.ca}}
}

\conferenceinfo{ISSAC '16,}{July 19--22, 2016, Waterloo, ON, Canada.}
\doi{http://dx.doi.org/10.1145/2930889.2930897}

\begin{document}

\maketitle

\begin{abstract}
The $p$-curvature of a system of linear differential equations in positive
characteristic $p$ is a matrix that measures how far the system is from having
a basis of polynomial solutions. We show that the similarity class of the
$p$-curvature can be determined without computing the $p$-curvature itself.
More precisely, we design an algorithm that computes the invariant factors of
the $p$-curvature in time quasi-linear in~$\sqrt p$. This is much less than
the size of the $p$-curvature, which is generally 
linear in~$p$. The new algorithm allows to answer a question originating from the study of the Ising model in statistical~physics. 
\end{abstract}

\vspace{-1mm}

\vspace{-1.5mm}

\medskip


\section{Introduction}\label{sec:intro}

Differential equations in positive characteristic~$p$ are important and
well-studied objects in mathematics~\cite{Honda81,Put95,Put96}. The main
reason is arguably one of Grothendieck’s (still unsolved)
conjectures~\cite{Katz72,Katz82,Andre04}, stating that a linear differential
equation with coefficients in $\mathbb{Q}(x)$ admits a basis of algebraic
solutions if and only if its reductions modulo (almost) all primes~$p$ admit a
basis of polynomial solutions modulo~$p$. Another motivation stems from the
fact that the reductions modulo prime numbers yield useful information about
the factorization of differential operators in characteristic zero.

To a linear differential equation in fixed characteristic~$p$, or more
generally to a system of such equations, is attached a simple yet very useful
object, the \emph{$p$-curvature}. Let~$\F_q$ be the finite field with~$q=p^a$
elements. The $p$-curvature of a system of linear differential equations with
coefficients in~$\F_q(x)$ is a matrix with entries in~$\F_q(x)$ that measures
the obstructions for the given system to possess a fundamental matrix of
polynomial solutions in~$\F_q[x]$. Its definition is remarkably simple,
especially at a higher level of generality: the $p$-curvature of a
differential module~$(M,\partial)$ of dimension~$r$ over~$\F_q(x)$ is the
``differential-Frobenius-map'' $\partial^p = \partial \circ \cdots
\circ\partial$ ($p$~times). When applied to the differential module
canonically attached with the system $Y' = A(x)Y$, the $p$-curvature
materializes into the $p$-th iterate $\partial_A^p$ of the map $\partial_A :
\F_q(x)^r \rightarrow \F_q(x)^r$ that sends $v$ to $v' - A v$, or more
concretely, into the matrix $A_p(x)$ of this map with respect to the canonical
basis of $\F_q(x)^r$. It is given as the term $A_p$ of the sequence $(A_i)_i$
of matrices in ${M}_{r}(\F_q(x))$ defined by

\smallskip\centerline{$A_1 = -A \quad \text{and} \quad A_{i+1} = A'_i - A \cdot A_i \quad \text{for} \quad i \geq 1.$}

\smallskip

{}From a computer algebra perspective, many effectivity questions naturally
arise. They primarily concern the algorithmic complexity of various operations
and properties related to the $p$-curvature: How fast can one compute $A_p$?
How fast can one decide its nullity? How fast can one determine its minimal
and characteristic polynomial? Apart the fundamental nature of these questions
from the algebraic complexity theory viewpoint, there are concrete motivations
for the efficient computation of the $p$-curvature, coming from various
applications, notably in enumerative combinatorics and statistical
physics~\cite{BoKa08b,BoKa08a,BBHMWZ08}.

We pursue the algorithmic study of the $p$-curvature, initiated
in~\cite{BoSc09,BoCaSc14,BoCaSc15}. In those articles, several questions were
answered satisfactorily, but a few other problems were left open. In summary, the
current state of affairs is as follows. First, the $p$-curvature~$A_p$ can be
computed in time $O(\log p)$ when $r=1$ and $\softO(p)$ when $r>1$. The soft-O
notation $\softO(\,)$ indicates that polylogarithmic factors in the argument
of $O(\,)$ are deliberately not displayed. These complexities match, up to
polylogarithmic factors, the generic size of~$A_p$. Secondly, one can decide
the nullity of~$A_p$ in time~$\softO(p)$ and compute its characteristic polynomial in time
$\softO(\sqrt{p})$. It is not known whether the exponent
$1/2$ is optimal for the last problem. In all these estimates, the
complexity (``time'') measure is the number of arithmetic operations $(\pm,
\times, \div)$ in the ground field~$\F_q$, and the dependence is expressed in
the main parameter~$p$ only. Nevertheless, precise estimates are also
available in terms of the other parameters of the input.

In the present work, we focus on the computation of all the invariant factors
of the $p$-curvature, and show that they can also be determined in
time~$\softO(\sqrt{p})$. Previously, this was unknown even for the minimal
polynomial of $A_p$ or for testing the nullity of~$A_p$. The fact that a sublinear cost could in principle be
achievable, although~$A_p$ itself has a total arithmetic size linear in $p$,
comes from the observation that the coefficients of the invariant factors of
$A_p$ lie in the subfield $\F_q(x^p)$ of $\F_q(x)$, in other words they are
very sparse.

To achieve our objective, we blend the methods used in our previous works~\cite{BoCaSc14} and~\cite{BoCaSc15}. 
The first key ingredient is the construction, for any point~$a$ in the algebraic closure of~$\F_q$ that is not a pole of~$A(x)$, of a matrix $Y_a$
with entries in $\ell = \F_q(a)$ which is similar to the evaluation $A_p(a)$
of the $p$-curvature at the point~$a$.
This construction comes from~\cite{BoCaSc15} and ultimately relies 
on the existence of a well-suited ring, of so-called \emph{Hurwitz series in $x-a$}, for which an analogue of the Cauchy–Lipschitz theorem holds for the system $Y'=A(x)Y$ around the (ordinary) point $x=a$. The matrix~$Y_a$ is the $p$-th coefficient of the fundamental matrix of Hurwitz series solutions of $Y'=A(x)Y$ at $x=a$.

The second key ingredient is a baby step / giant step algorithm that computes~$Y_a$ in $\softO(\sqrt{p})$ operations in $\ell$ via fast matrix factorials. Finally, we recover the invariant factors of~$A_p$ from those of the matrices~$Y_a$, for a suitable number of values~$a$.
The main difficulty in this interpolation process is that there exist badly behaved points~$a$ for which the invariant factors of $A_p(a)$ are not the evaluations at~$a$ of the invariant factors of $A_p(x)$. The remaining task is then to bound the number of unlucky evaluation points~$a$. The key feature allowing a good control on these points, independent of~$p$, is the fact that the invariant factors of $A_p(x)$ have coefficients in $\F_q(x^p)$. 

\smallskip\noindent{\bf Relationship to previous work.}  There exists
a large body of work concerning the computation of so-called Frobenius
forms of matrices (that is, the list of their invariant factors,
possibly with corresponding transformation matrices), and the related
problem of Smith forms of polynomial matrices. The specificities of
our problem prevent us from applying these methods directly; 
however, our work is related to several of these previous results.

Let~$\omega$ be a feasible exponent for matrix multiplication. 
The best deterministic algorithm known so far
for the computation of the Frobenius form of an $n \times n$ matrix over a
field~$k$
is due to Storjohann~\cite{Storjohann01}. This algorithm has running time $O(n^\omega \log(n)
\log\log(n))$ operations in $k$.
We will use it to compute
the invariant factors of the matrices $Y_a$ above.  Las Vegas
algorithms were given by Giesbrecht~\cite{Giesbrecht95},
Eberly~\cite{Eberly00} and Pernet and Storjohann~\cite{PeSt07}, the
latter having expected running time $O(n^\omega)$ over sufficiently
large fields.

The case of matrices with integer or rational entries has attracted a
lot of attention; this situation is close to ours, with the bit size
of integers playing a role similar to the degree of the 
entries in the $p$-curvature. Early work goes back to algorithms of
Kaltofen {\it et al.}~\cite{KaKrSa87,KaKrSa90} for the Smith form of
matrices over $\Q[x]$, which introduced techniques used in several
further algorithms, such as the Las Vegas algorithm by Storjohann and
Labahn~\cite{StLa97}. Giesbrecht's PhD thesis~\cite{GiesbrechtPhD}
gives a Las Vegas algorithm with expected cost
$O\tilde{~}(n^{\omega+2} d)$ for the Frobenius normal form of an $n
\times n$ matrix with integer entries of bit size $d$; Storjohann
and Giesbrecht substantially improved this result in~\cite{GiSt02},
with an algorithm of expected cost $O\tilde{~}(n^4 d + n^3 d^2)$.
The best Monte Carlo running time known to us is
$O\tilde{~}(n^{2.698} d)$, by Kaltofen and Villard~\cite{KaVi04}.

In the latter case of matrices with integer coefficients, a common
technique relies on reduction modulo primes, and a main source of
difficulty is to control the number of ``unlucky'' reductions. We
pointed out above that this is the case in our algorithm as well.  In
general, the number of unlucky primes is showed to be $O\tilde{~}(n^2
d)$ in~\cite{GiesbrechtPhD}; in our case, the degree $d$ of the
entries grows linearly with $p$, but as we said above, we can
alleviate this issue by exploiting the properties of the
$p$-curvature. Storjohann and Giesbrecht proved in~\cite{GiSt02} that
a candidate for the Frobenius form of an integer matrix can be verified using only
$O\tilde{~}(nd)$ primes; it would be most interesting to adapt this 
idea to our situation.

\medskip\noindent{\bf Structure of the paper.} In
Section~\ref{sec:invfact}, we recall the main theoretical properties of the
invariant factors of a polynomial matrix, and study their behavior under specialization. We obtain bounds on bad evaluation points, and use them to design (deterministic and probabilistic) evaluation-interpolation algorithms 
for computing the invariant factors of a polynomial matrix.
Section~\ref{sec:pcurv} is devoted to the design of our main algorithms for the similarity class of the $p$-curvature, with deterministic and probabilistic versions for both the system case and the scalar case. 
Finally, Section~\ref{sec:application} presents an application of our algorithm, that allows to answer a question coming from theoretical physics.

\medskip\noindent{\bf Complexity basics.} 
We use standard complexity notation, such as $\omega$ for the exponent of matrix multiplication.
The best known upper bound is $\omega <2.3729$ from~\cite{LeGall14}. 
Many arithmetic operations on
univariate polynomials of degree $d$ in $k[x]$ can be performed in $\softO(d)$
operations in the field~$k$: addition, multiplication, shift, interpolation,
\emph{etc}, the key to these results being fast polynomial
multiplication~\cite{Schoenhage77,CaKa91,HaHoLe14}.  A general
reference for these questions in~\cite{GaGe03}. 


\section{\hspace{-0.18cm} Computing invariant factors of \\ \hspace{-0.18cm}  special polynomial matrices}
\label{sec:invfact}

\subsection{Definition and classical facts}

We recall here some basic facts about invariant factors of matrices 
defined over a field. We fix for now a field~$K$, and a matrix 
$M \in {M}_n(K)$. For a monic polynomial $P = T^d - \sum_{i=0}^{d-1} a_i 
T^i \in K[T]$, let $M_P$ denote its companion matrix:
$$M_P = \left(\begin{matrix}
 &   &  & a_0 \\
1 &   &  & a_1 \\
 &  \ddots &  & \vdots \\
 &   & 1 & a_{d-1}
\end{matrix}\right).$$

A well-known theorem~\cite[Th.~9, Ch.~VII]{Gantmacher59} asserts that there exist a unique sequence of 
monic polynomials $I_1, \ldots, I_n$ for which $I_j$ divides 
$I_{j+1}$ for all $j$ and $M$ is similar to a block diagonal matrix 
whose diagonal entries are $M_{I_1}, \ldots, M_{I_n}$. The $I_j$'s
are called the \emph{invariant factors} of $M$. We emphasize that, with
our convention, there are always $n$ invariant factors but some of
them may be equal to $1$, in which case the corresponding companion
matrix is the empty one. Under this normalization, the $j$-th invariant 
factor $I_j$ can be obtained as $I_j = G_j/G_{j-1}$, where $G_j$ is
the gcd of the minors of size $j$ of the 
matrix $T \text{I}_n - M$, where $\text{I}_n$ stands for the identity 
matrix of size $n$.
The invariant factors are closely related to the characteristic
polynomial; indeed, we have 
\begin{equation}
\label{eq:Ijchi}
I_1 \cdot I_2 \cdots I_n = G_n = \det(T \text{I}_n - M).
\end{equation}

Given some irreducible polynomial $P$ in~$K[T]$, we 
consider the sequence (of integers):
\begin{equation}
\label{eq:seqdim}
e \mapsto d_{P,e} = \frac{\dim_K \ker P^e(M)}{\deg P}.
\end{equation}
It turns out that this sequence completely determines the $P$-adic valuation 
of the invariant factors. Indeed, denoting by $v_j$ the $P$-adic
valuation of $I_j$, we have the relations:
\begin{align}
d_{P,e} & = \sum_{j=1}^n \min(e, v_j), \label{eq:dim} \smallskip \\
d_{P,e} - d_{P,e-1} & = \text{Card} \{ j \: | \: v_j\:{\geq}\:e \} \label{eq:diffdim}
\end{align}
from which the $v_j$'s can be recovered without ambiguity since they
form a nondecreasing sequence. It also follows from the above
formula that the sequence $e \mapsto d_{P,e}$ is concave and eventually 
constant. Its final value is the dimension of the characteristic subspace
associated to $P$ and it is reached as soon as $e$ is greater than or 
equal to $v_n$.

\subsection{Behaviour under specialization}

Let $k$ be a perfect field of characteristic $p$. We consider~a matrix $M(x)$ with coefficients in 
$k[x]$. For an element $a$ lying in a finite extension $\ell$ of $k$, 
we denote by $M(a)$ the image of $M(x)$ under the mapping $k[x] \to 
\ell$, $x \mapsto a$. Our aim is to compare the invariant factors of
$M(x)$ and those of $M(a)$.

We introduce some notation. Let $I_1(x,T), \ldots, I_n(x,T)$ be the
invariant factors of $M(x)$.
It follows from the relation~\eqref{eq:Ijchi} that they all lie in
$k[x,T]$. We can therefore evaluate them at $x=a$ for each element
$a \in \ell$ as above and get this way univariate polynomials with
coefficients in $\ell$. Let $I_1(a,T), \ldots, I_n(a,T)$ be these
evaluations. We also consider the invariant factors of $M(a)$ and
call them $I_{1,a}(T), \ldots, I_{n,a}(T)$. We furthemore define
\begin{align*}
G_j(x,T) & = I_1(x,T) \cdot I_2(x,T) \cdots I_j(x,T) \\
\text{and}\quad
G_{j,a}(T) & = I_{1,a}(T) \cdot I_{2,a}(T) \cdots I_{j,a}(T).
\end{align*}

The characterization of the $G_j$'s in term of minors yields:
\vspace{-1ex}
\begin{lem}
\label{lem:divide}
For all $a \in \ell$ and all $j \in \{1, \ldots, n\}$, 
the polynomial $G_j(a,T)$ divides $G_{j,a}(T)$ in $\ell[T]$.
\end{lem}

Let $P_1(x,T), \ldots, P_s(x,T)$ be the irreducible factors of
the characteristic polynomial $\chi(x,T)$ of $M(x)$, and let us write $\chi^\textrm{sep}(x,T)$ for $P_1(x,T) \cdots P_s(x,T)$.
For all $1\leq i\leq s$ and $1\leq j\leq n$, let $e_{i,j}$ be
the multiplicity of $P_i(x,T)$ in $I_j(x,T)$.

\begin{prop}
\label{prop:goodspecialization}
We assume $\chi^\textrm{\emph{sep}}(a,T)$ is separable and 
$$\dim_{k(x)} \ker P_i(x,M(x))^{e_{i,j}+1} 
= \dim_{\ell} \ker P_i(a,M(a))^{e_{i,j}+1}$$
for all $i$ and for all $j<n$. Then $I_j(a,T) = I_{j,a}(T)$ for all~$j$.
\end{prop}

\begin{proof}
The equality of dimensions is also true for $j=n$, as their sum on both sides equals~$n$ (using separability) and these dimensions can only increase by specialization.
Let $d_{P_i,e}$ be the sequence defined by Eq.~\eqref{eq:seqdim} with 
respect to the irreducible polynomial $P_i(x,T)$ and the matrix~$M(x)$. 
We define similarly for each irreducible factor $P(T)$ of $P_i(a,T)$ the 
sequence $d_{P,e}$ corresponding to the polynomial $P(T)$ and the matrix 
$M(a)$. 
We claim that it is enough to prove that $d_{P_i,e} = d_{P,e}$ for all 
$e$, $i$ and all irreducible divisors $P(T)$ of $P_i(a,T)$. Indeed,
by Eq.~\eqref{eq:diffdim}, such an equality would imply:
\begin{equation}
\label{eq:valP1}
v_{P(T)}(I_{j,a}(T)) = e_{i,j}
\end{equation}
provided that $P(T)$ is an irreducible divisor of $P_i(a,T)$, and where 
$v_{P(T)}$ denotes the $P(T)$-adic valuation. On the other hand, 
still assuming that $P(T)$ is an irreducible divisor of $P_i(a,T)$,
it follows from the definition of the $e_{i,j}$'s that:
\begin{equation}
\label{eq:valP2}
v_{P(T)}(I_j(a,T)) \geq e_{i,j}
\end{equation}
and that the equality holds if and only if $P(T)$ does not divide any
of the $P_{i'}(a,T)$ for $i' \neq i$.
Comparing characteristic polynomials, we know moreover that 
$\sum_{j=1}^n v_{P(T)}(I_{j,a}(T)) = \sum_{j=1}^n v_{P(T)}(I_j(a,T))$.
Combining this with \eqref{eq:valP1} and \eqref{eq:valP2}, we find
that the $P_i(a,T)$'s are pairwise coprime and finally get 
$I_j(a,T) = I_{j,a}(T)$ for $1 \leq j \leq n$, as wanted.

Until the end of the proof, we fix the index $i$ and reserve the letter 
$P$ to denote an irreducible divisor of $P_i(a,T)$. For a fixed integer 
$e$, denote by $j_0$ the greatest index $j$ for which 
$v_{P(T)}(I_{j,a}(T)) < e$ and observe that Eq.~\eqref{eq:dim} can be 
rewritten
$d_{P,e} = e \cdot \big(n-j_0\big) + v_{P(T)}\big(G_{j_0,a}(T)\big)$.
Using Lemma~\ref{lem:divide}, we derive
$d_{P,e} \geq e \cdot \big(n-j_0\big) + v_{P(T)}\big(G_{j_0}(a,T)\big) 
\geq d_{P_i,e}$ for all $P$ and $e$.
Eq.~\eqref{eq:diffdim} now implies that the indices $e$ for which 
$d_{P_i,e} - d_{P_i,e-1} >
d_{P_i,e+1} - d_{P_i,e}$ are exactly the $e_{i,j}$'s ($1 \leq j \leq n$).
Using concavity, we then observe that it is enough to check that 
$d_{P_i,e} = d_{P,e}$ for indices $e$ of the form $e_{i,j} + 1$. For 
those $e$, we have by assumption:
$$\begin{array}{r@{\,\,}l}
\sum_P \deg P \cdot d_{P,e} & = \dim_\ell \ker P_i(a,M(a))^e \smallskip \\
& = \dim_{k(x)} \ker P_i(x,M(x))^e \smallskip \\
& = \deg_T P_i \cdot d_{P_i,e} 
= \sum_P \deg P \cdot d_{P_i,e}
\end{array}$$
and thus $d_{P,e} = d_{P_i,e}$ for all $P$ because the inequalities
$d_{P,e} \geq d_{P_i,e}$ are already known.
\end{proof}

\subsection{A bound on bad evaluation points}
\label{ssec:boundbad}
Let $M(x)$ be a square matrix of size $n$ with coefficients in~$k[x]$. 
We set $X = x^p$ and assume that:

\noindent
(i) the entries of $M(x)$ have degree at most $pm$ (for a $m\in\mathbb{N}$), 

\noindent
(ii) $M(x)$ is similar to a matrix with coefficients in $k(X)$.

\smallskip

We are going to bound the number of values of $a$ for which the 
invariant factors of $M(x)$ do not specialize correctly at $x=a$. 
Similar discussions appear is Section~4 of Giesbrecht's 
thesis \cite{GiesbrechtPhD} in the (more complicated) case of integer 
matrices. Our treatment is nevertheless rather different in many 
places.

\medskip

\noindent
{\bf The basic bound.}
By assumption~(ii), the characteristic polynomial $\chi(x,T)$ lies 
in the subring $k[X,T]$ of $k[x,T]$.

\begin{lem}
The invariant factors $I_j(x,T)$ all belong to $k[X,T]$. Their degree 
with respect to $X$ is at most $mn$.
\end{lem}

\begin{proof}
By assumption~(i), $\chi(x,T)$ is a polynomial in $x$ of degree at most 
$pmn$. It then follows from Eq.~\eqref{eq:Ijchi} that the $I_j(x,T)$'s 
are polynomials in $x$ of degree at most $pmn$ as well. Now, the 
assumption~(ii) ensures that the $I_j(x,T)$'s actually lie in 
$k(X)[T]$. This completes the proof.
\end{proof}

\begin{lem}\label{lem:badpoints}
We assume that $p>n$.
There are at most $\deg_X \chi(x,T) \cdot (2n-1)$ 
points $a\in k$ such that at least one of the $P_i(a,T)$'s is not separable.
\end{lem}

\begin{proof}
We have that
$\deg_X \chi^\textrm{sep}(x,T) \leq \deg_X \chi(x,T)$ and
$\deg_T \chi^\textrm{sep}(x,T) \leq n$, since $\chi^\textrm{sep}$ divides $\chi$.
Denote by $D(x)$ the discriminant of $\chi^\textrm{sep}(x,T)$ with respect to~$T$. Its degree in $X$ is at most $\deg_X \chi(x,T)\cdot (2n-1)$,
and the assumption $p > n$ implies that $D(x)$ is not identically zero. For any $a\in k$ such that $D(a^p) \neq 0$, the polynomial $\chi^\textrm{sep}(a^p,T)$ is separable, and the same holds for the $P_i(a^p,T)$'s.
Noting that $k$ is perfect, the conclusion holds.
\end{proof}

\begin{prop}
\label{prop:boundbad}
We assume $p > n$.
Let $a_1, \ldots, a_N$ be elements in a separable closure of $k$
which are pairwise non conjugate over $k$. We assume that for each
$i \in \{1, \ldots, N\}$, there exists $j \in \{1, \ldots,n\}$
with $I_j(a_i,T) \neq I_{a_i,j}(T)$.
Then:
$$\sum_{i=1}^N \deg(a_i) \leq 4mn\cdot(n-1) + mn\cdot(2n-1)
$$
where $\deg(a_i)$ denotes the 
algebraicity degree of $a_i$ over $k$.
\end{prop}

\begin{proof}
We use the criteria of Proposition~\ref{prop:goodspecialization}.
We start by putting away the values of $a$ for which 
at least one of the $P_i(a,T)$'s is not separable.
By Lemma~\ref{lem:badpoints}, there are at most $mn\cdot(2n-1)$ such values.
We then have to bound from
above the values of $a$ such that the equalities:
$$\dim_{k(x)} \ker P_i(x,M(x))^e
= \dim_{\ell} \ker P_i(a,M(a))^e$$
may fail for some $i$ and some exponent 
$e=e_{i,j} + 1$ for some~$j$. 

Let us fix such a pair $(i,e)$. Set $N(x) = P_i(x,M(x))^e$ for 
simplicity. By assumption~(i), the entries of $N(x)$ have degree at
most $p m_{i,e}$ with
$m_{i,e} = e \cdot \big( m \deg_T P_i + \deg_X P_i \big)$.
On the other hand, we deduce from assumption~(ii) that the $P_i(x,T)$'s 
all lie in $k[X,T]$ 
and, as a consequence, that $N(x)$ is similar to a matrix with 
coefficients in $k(X)$. Define $d = \dim_{k(x)} \ker N(x)$. The equality 
$\dim_{\ell} \ker N(a) = d$ then fails if and only if the minors of 
$N(x)$ of size $n-d$ all vanish at $x=a$, i.e., if and only if the 
gcd $\Delta(x)$ of these minors is divisible by the minimal polynomial
of $a$ over $k$, say $\pi_a(x)$. Noting that
$\Delta(x) \in k[X]$, the latter condition is also equivalent to the
fact that $\pi_a(x)^p$ divides $\Delta(x)$ in the ring $k[X]$. This can
be possible for at most $\deg_X \Delta(x) \leq (n-d) m_{i,e} \leq 
(n-1) m_{i,e}$ values of $a$.

Therefore, if $a_1, \ldots, a_N$ are pairwise non-conjugate
``unlucky values'' of $a$, the
sum appearing in the statement of the proposition is bounded from 
above by:
\begin{align*}
\textstyle (n-1) \sum_{i,e} m_{i,e} 
& \textstyle = m(n-1) \sum_{i,e} e \deg_T P_i \\
& \hspace{1cm}\textstyle + (n-1) \sum_{i,e} e \deg_X P_i.
\end{align*}
We notice that, when $i$ remains fixed, the number of exponents of the 
form $e_{i,j}+1$ ($1 \leq j < n$) is  bounded from above by 
$e_{i,n} + 1$. The sum of these exponents is then at most:
$$\textstyle \big(\sum_{j=1}^{n-1} e_{i,j}\big) + e_{i,n} + 1 =
e_i + 1 \leq 2 e_i,$$
where $e_i$ denotes the multiplicity of the factor $P_i(x,T)$ in 
the characteristic polynomial $\chi(x,T)$. Our bound then becomes
$2m(n-1) \deg_T \chi + 2(n-1) \deg_X \chi$.
Using $\deg_T \chi = n$ and $\deg_X \chi \leq mn$ yields the 
bound.
\end{proof}

\noindent
{\bf A refinement.}
For the applications we have in mind, we shall need a refinement 
of Proposition~\ref{prop:boundbad} under the following hypothesis
depending on a parameter $\mu \in \N$:

\smallskip\centerline{$(\mathbf{H}_\mu)$: the polynomial $\chi$ has
degree at most $p\mu$ w.r.t $x$.}

\smallskip\noindent We observe that $(\mathbf{H}_\mu)$ is fulfilled when
$M(x)$ is a companion matrix whose entries are polynomials of degree
at most $p\mu$. 

\begin{prop}
\label{prop:boundbad2}
Under the assumptions of Prop.~\ref{prop:boundbad} and
the additional hypothesis $(\mathbf{H}_\mu)$, we have:
$$\sum_{i=1}^N \deg(a_i) \leq 2\mu\cdot (2n-1) + \mu\cdot (2n-1).$$
\end{prop}

\begin{proof}
Let $P(x,T)$ be any bivariate polynomial with coefficients in $k$. 
Set $N(x) = P(x,M(x))$ and let $\delta(x)$ denote the gcd of the minors
of size $s$ (for some integer $s$) of $N(x)$. We claim that:
\begin{equation}
\label{eq:bounddegdelta}
\deg_x \delta(x) \leq p \mu \cdot \deg_T P + s \cdot \deg_x P
\end{equation}
To prove the claim, we consider the Frobenius normal form $\tilde M(x)$ 
of $M(x)$ and set $\tilde N(x) = P(x, \tilde N(x))$.
Observe that any minor of $\tilde M(x)$ vanishes or has 
the shape $\pm c_1(x) \cdots c_n(x)$ where $c_j(x)$ is a 
coefficient of $I_j(x,T)$ for all $j$. Noting that 
$\deg_x I_1 + \cdots + \deg_x I_n = \deg_x \chi \leq p\mu$,
we derive that \emph{all} the minors of $\tilde M(x)$ have degree at 
most $p\mu$. Now write $P(x,T) = \sum_{j=0}^{\deg_T\!P}
a_j(x) T^j$ where the $a_i(x)$'s lie in $k[x]$.
Let $\tilde f$ denote the $k[x]$-linear endomorphism of $k[x]^n$ 
attached to the matrix $\tilde M(x)$. Set $\tilde g = P(x,\tilde f)$; 
it clearly corresponds to $\tilde N(x)$. Given a vector space $E$ and $s$ linear
endomorphisms $u_1, \ldots, u_s$ of $E$, let us agree to define
$u_1 \wedge \cdots \wedge u_s$ as
$$\begin{array}{rcl}
\quad E^{\otimes s} & \to & {\textstyle \bigwedge^s} E \\
x_1 \otimes \cdots \otimes x_s & \mapsto & u_1(x_1) \wedge \cdots \wedge u_s(x_s).
\end{array}$$
where $\bigwedge^s E$ is here defined as a quotient of $E^{\otimes s}$.
Expanding the exterior product $\bigwedge^s \tilde g$, we get:
\begin{equation}
\label{eq:wedgeg}
{\textstyle \bigwedge^s \tilde g} = \sum_{i_1, \ldots, i_s = 0}^{\deg_T P}
a_{i_1}(x) \cdots a_{i_s}(x) \cdot \tilde f^{i_1} \wedge \cdots \wedge 
\tilde f^{i_s}.
\end{equation}
Moreover, assuming for simplicity that $i_1 \leq i_2 \leq \cdots
\leq i_s$ and letting $i_0 = 0$ by convention, we can write:
$$\tilde f^{i_1} \otimes \cdots \otimes \tilde f^{i_s} =
\bigcirc_{j=0}^s 
\big[
({\textstyle \bigotimes^j \text{id}}) \otimes
({\textstyle \bigotimes^{s-j} \tilde f})^{i_j-i_{j-1}} \big],$$
where $\bigcirc$ denotes the composition of the above 
(pairwise commuting) maps. 
We get that the entries of the matrix (in the canonical
basis) of $\tilde f^{i_1} \wedge 
\cdots \wedge \tilde f^{i_s}$ all have degree at most $p \mu \cdot i_s$. 
The same argument demonstrates that the degrees of the 
entries of the above matrix are not greater than:
$$p\mu \cdot \max (i_1, \ldots, i_s) \leq p \mu \cdot \deg_T P$$
when we no longer assume that the $i_j$'s are sorted by nondecreasing 
order. Therefore, back to Eq.~\eqref{eq:wedgeg}, we find that the 
entries of $\bigwedge^s \tilde N(x)$ have degree at most $p \mu \cdot 
\deg_T P + s \cdot \deg_x P$. It is then also the case of its trace, 
which is the same as the trace of $\bigwedge^s 
N(x)$ since $N(x)$ and $\tilde N(x)$ are similar. This finally implies 
the claimed inequality~\eqref{eq:bounddegdelta} because $\delta(x)$
has to divide this trace.

The Proposition now follows by inserting the above input
in the proof of Proposition~\ref{prop:boundbad}.
\end{proof}

\subsection{Algorithms}

We keep the matrix $M(x)$ satisfying the assumptions (i) and (ii)
of \S\ref{ssec:boundbad}.
From now on, we assume that the only access we have to the matrix $M(x)$ 
passes through a black box {\tt invariant\_factors\_at${}_{M(x)}$} that 
takes as input an element~$a$ lying in a finite extension $\ell$ of $k$ 
and outputs instantly the invariant factors $I_{j,a}(T)$ of the matrix 
$M(a)$. Our aim is to compute the invariant factors of $M(x)$.
We will propose two possible approaches: the first one is deterministic
but rather slow although the second one is faster but probabilistic and
comes up with a Monte-Carlo algorithm which may sometimes output wrong
answers.

Throughout this section, the letter $D$ refers to \emph{a priori} 
upper bound on the $X$-degree of the characteristic polynomial of $M(x)$. 
One can of course always take $D = mn$ but better bounds might be 
available in particular cases. Similarly we reserve the letter $F$
for an upper bound on the sum of degrees of ``unlucky evaluation
points''. Proposition~\ref{prop:boundbad} tells us that $mn(6n-5)$ 
is always an acceptable value for~$F$. Remember however that this value can 
be lowered to $3 \mu (2n-1)$ under the hypothesis $(\mathbf{H}_\mu)$
thanks to Proposition~\ref{prop:boundbad2}. We will always assume
that $F \geq D$.

For simplicity of exposition, we assume from now on that $k = \F_q$ is 
a finite field of cardinality $q$ (it is more difficult and the case of
most interest for us).

\medskip

\noindent
{\bf Deterministic.}
The discussion of \S \ref{ssec:boundbad} suggests the following algorithm 
whose correctness follows directly from the definition of $F$ together 
with the assumption $F \geq D$.

\noindent\hrulefill

\noindent {\bf Algorithm} {\tt invariant\_factors\_deterministic}

\noindent{\bf Input:} $M(x)$ satisfying (i) and (ii), $D$, $F$
with $F \geq D$

\noindent{\bf Output:} The invariant factors of $M(x)$

\smallskip\noindent 1.\ %
Construct an extension $\ell$ of $\F_q$ of degree $F+1$

\noindent \hphantom{1.\ }%
and pick an element $a \in \ell$ such that $\ell = \F_q[a]$

\noindent \hphantom{1.\ }%
{\sc Cost:} $\softO(F)$ operations in $\F_q$

\smallskip\noindent 2.\ 
$I_{1,a}(T), \ldots, I_{n,a}(T) = \texttt{invariant\_factors\_at}{}_{M(x)}(a)$

\smallskip\noindent 3.\ 
{\bf for} $j=1,\ldots,n$

\smallskip\noindent 4.\ \hspace{3mm}%
Find $I_j(x,T)$ of degree $\leq D$ s.t. $I_j(a,T) = I_{j,a}(T)$

\smallskip\noindent 5.\ 
{\bf return} $I_1(x,T), \ldots, I_n(x,T)$

\vspace{-1ex}\noindent\hrulefill

\begin{prop}
\label{prop:costdeterministic}
The algorithm above requires only one call to the black 
box \texttt{invariant\_factors\_at}${}_{M(x)}$ with an input of
degree exactly $F+1$.
\end{prop}

\noindent
{\bf Probabilistic.}
We now present a Monte-Carlo algorithm:

\noindent\hrulefill

\noindent {\bf Algorithm} {\tt invariant\_factors\_montecarlo}

\noindent{\bf Input:} $M(x)$ s.t. (i) and (ii), $\varepsilon
\in (0,1)$,
$D$, $F$ with $F \geq D$

\noindent{\bf Output:} The invariant factors of $M(x)$

\smallskip\noindent 1.\ %
Find the smallest integer $s$ such that:
\begin{equation}
\label{eq:boundeps}
2 \cdot \frac{(D{+}s{+}1)^2}{s (q^s - 2F)} + 
\frac 1 2 \cdot \Big(\frac {4F}{q^s} \Big)^{\!(D{-}2)/s} \:\leq \:\varepsilon
\end{equation}

\noindent \hphantom{1.\ }%
and set $K = \lceil \frac{3D}s \rceil$ and 
$k = \lceil \frac{D+1}s \rceil$.

\medskip\noindent 2.\ 
{\bf for} $i=1,\ldots,K$

\smallskip\noindent 3.\ \hspace{3mm}%
pick at random $a_i \in \F_{q^s}$ s.t. $\F_{q^s} = \F_q[a_i]$

\noindent \hphantom{3.\ \hspace{3mm}}%
{\sc Cost:} $\softO(s)$ operations in $\F_q$

\smallskip\noindent 4.\ \hspace{3mm}%
$I_{1,i}(T), \ldots, I_{n,i}(T) = \texttt{invariant\_factors\_at}{}_{M(x)}(a_i)$

\medskip\noindent 5.\ 
{\bf for} $j=1,\ldots,n$

\smallskip\noindent 6.\ \hspace{3mm}%
$d_j = \max_i \deg(I_{1,i} (T) \cdot I_{2,i} (T) \cdots I_{j,i}(T))$

\smallskip\noindent 7.\ \hspace{3mm}%
select $I \subset \{1, \ldots, K\}$ of cardinality $k$ s.t.

\noindent \hphantom{7.\ \hspace{3mm}}%
(i)~$\deg(I_{1,i}(T)\cdot I_{2,i}(T) \cdots I_{j,i}(T)) = d_j$ for all $i \in I$

\noindent \hphantom{7.\ \hspace{3mm}}%
(ii)~the $a_i$ are pairwise non conjugate for $i \in I$

\noindent \hphantom{7.\ \hspace{3mm}}%
{\sc Remark:} if such $I$ does not exist, raise an error

\smallskip\noindent 8.\ \hspace{3mm}%
compute $I_j \in \F_q[X,T]$ of $X$-degree $\leq D$ s.t. 

\noindent \hphantom{8.\ \hspace{3mm}}%
$I_j(a_i,T) = I_{j,i}(T)$ for all $i \in I$

\noindent \hphantom{7.\ \hspace{3mm}}%
{\sc Cost:} $\softO(D)$ operations in $\F_q$

\medskip\noindent 9.\ 
{\bf return} $I_1(x,T), \ldots, I_n(x,T)$

\vspace{-1ex}\noindent\hrulefill

\smallskip

\begin{prop}
\label{prop:costmontecarlo}
We have $s \in O(\log \frac{FD}\varepsilon)$. Moreover:

\smallskip

\noindent $\bullet$
\emph{Correctness:}
Algorithm {\tt invariant\_factors\_montecarlo}
fails or returns a wrong answer with probability at most 
$\varepsilon$.

\smallskip

\noindent $\bullet$
\emph{Complexity:}
It performs
$\lceil \frac{3D}s \rceil$ calls to the black box
with inputs of degree $s$ and
$\softO(n(D + \log \frac F \varepsilon))$
operations in $\F_q$.
\end{prop}

\begin{proof}
The first assertion is left to the reader.
Let $\mathcal A$ be the set of elements $a$ of $\F_{q^s}$ such that 
$\F_q[a] = \F_{q^s}$. It is an easy exercise to prove that $\mathcal A$ 
has at least $\frac{q^s} 2$ elements (the bound is not sharp). Let 
$\mathcal C_1, \ldots, \mathcal C_C$ be the conjugacy classes (under the Galois action) in
$\mathcal A$. Remark that each $\mathcal C_i$ has by definition $s$ elements,
so that $C \geq \frac{q^s}{2s}$. We say that a conjugacy class is
\emph{bad} if it contains one element $a$ for which $I_j(a,T) \neq 
I_{a,j}(T)$ for some $j$. Otherwise, we say that it is \emph{good}.
Let $B$ (resp. $G$) be the number of bad (resp. good) classes. We 
have $B+G = C$ and $B \leq \frac F s$ by definition of $F$.

The algorithm {\tt invariant\_factors\_montecarlo} succeeds if
there exist at least $k$ indices $i$ for which the corresponding 
$a_i$'s lie in pairwise distinct good classes. This happens with
probability at least:
$$\frac 1{C^K} \cdot 
\sum_{j=k}^K {\textstyle \binom{K}{j}} \cdot
 G (G-1) \cdots (G-k+1) \cdot G^{j-k} \cdot B^{K-j}.$$ 
(The above formula gives the probability that the \emph{first} $k$ good 
classes are pairwise distinct, which is actually stronger than what we 
need.) The above quantity is at least equal to
\begin{align*}
\left(1 - \frac k G\right)^{\!k}\: \cdot
\left(1 \: - \:
\sum_{j=0}^{k-1} {\textstyle \binom{K}{j}} \cdot \Big(\frac G C\Big)^j 
\cdot \Big(\frac B C\Big)^{K-j}\right).
\end{align*}
Moreover for $j \leq k-1$, we have:
\begin{align*}
\Big(\frac G C\Big)^j
\cdot \Big(\frac B C\Big)^{K-j}
& \leq \Big(\frac {BG}{C^2}\Big)^j \!\cdot\! \Big(\frac B C\Big)^{K-2j}
\leq \frac 1{2^{2j}} \cdot \Big(\frac {2F}{q^s} \Big)^{\!K{-}2j}\\
& \leq \frac 1{2^K} \cdot \Big(\frac {4F}{q^s} \Big)^{\!K{-}2j}
\leq \frac 1{2^K} \cdot \Big(\frac {4F}{q^s} \Big)^{\!(D{-}2)/s}.
\end{align*}
Therefore the probability of success is at least:
$$\left(1 - \frac k G\right)^{\!k}\: \!\cdot\!
\left(1 \: - \:
\frac 1 2 \cdot \Big(\frac {4F}{q^s} \Big)^{\!(D{-}2)/s}\right).$$
Using $k \leq \frac{D+s+1}s$ and $G \geq 
\frac{q^2-2F}{2s}$, we find that the probability of failure is at most
the LHS of Eq.~\eqref{eq:boundeps}. The correctness is proved.
As for the complexity, the results are obvious.
\end{proof}

\section{Computing invariant factors \\ of the p-curvature}\label{sec:pcurv}

Throughout this section, we fix a finite field $k = \F_q$ of cardinality 
$q$ and characteristic $p$.  
We endow the field of rational functions $k(x)$ with the natural derivation $f 
\mapsto f'$.

\subsection{The case of differential modules}

We recall that a differential module over $k(x)$ is $k(x)$-vector space
$M$ endowed with an additive map $\partial : M \to M$ satisfying the
following Leibniz rule:
$$\forall f \in k(x), \, \forall m \in M, \quad 
\partial (fm) = f' \cdot m + f \cdot \partial(m).$$
The \emph{$p$-curvature} of a differential module $M$ is the mapping
$\partial^p = \partial \circ  \cdots \circ \partial$ ($p$ times).
Using the fact that the $p$-th derivative of any $f \in k(x)$ vanishes, 
we derive from the Leibniz relation above that $\partial^p$ is 
$k(x)$-linear endomorphism of $M$.
It follows moreover from~\cite[Remark~4.5]{BoCaSc15} that $\partial^p$ is 
defined over $k(x^p)$, in the sense that there exists a $k(x)$-basis of 
$M$ in which the matrix of $\partial^p$ has coefficients in $k(x^p)$. In 
particular, all the invariant factors of the $p$-curvature have their coefficients in $k(x^p)$.

\medskip

\noindent
{\bf Statement of the main Theorem.}
From now on, we fix a differential module $(M, \partial)$. We assume that 
$M$ is finite dimensional over $k(x)$ and let $r$ denote its dimension. 
We pick $(e_1, \ldots, e_r)$ a basis of $M$ and let $A$ denote the 
matrix of $\partial$ with respect to this basis. We write $A = \frac 
1{f_A} \tilde A$ where $f_A$ and the entries of $\tilde A$ all lie
in $k[x]$. Let $d$ be an upper bound on the degrees of all these
polynomials.
The aim of this section is to design fast deterministic and probabilistic
algorithms for computing the invariant factors of the $p$-curvature of $(M, \partial)$.
The following Theorem summarizes our results.

\begin{theo}
\label{theo:pcurvmodules}
We assume $p > r$.

\smallskip

\noindent
1. There exists a deterministic algorithm that computes the
invariant factors of the $p$-curvature of $(M, \partial)$ within
$$\softO\big(d^{\hspace{0.2mm}\omega+\frac 3 2} r^{\omega+2} \sqrt{p}\big)$$
operations in $k = \F_q$.

\smallskip

\noindent
2. Let $\varepsilon \in (0,1)$.
There exists a Monte-Carlo algorithm computing the
invariant factors of the $p$-curvature of $(M, \partial)$~in
$$\softO\big(d^{\hspace{0.2mm}\omega + \frac 1 2} r^{\omega} \cdot(dr - \log \varepsilon) \cdot \sqrt{p}\big)$$
operations in $k = \F_q$.
This algorithm returns a wrong answer with probability at most
$\varepsilon$.
\end{theo}

In what follows, we will use the notation $A_p(x)$ for the matrix of the 
$p$-curvature of $(M, \partial)$ with respect to the distinguished basis 
$(e_1, \ldots, e_r)$. Given an element $a$ lying in a 
finite extension $\ell$ of $k$, we denote by $A_p(a) \in {M}_r(\ell)$ 
the matrix deduced from $A_p$ by evaluating it at $x=a$.

\medskip

\noindent
{\bf The similarity class of $A_p(a)$.}
Let $S$ be an irreducible polynomial over $k$. Set $\ell = k[u]/S$ and 
let $a$ denote the image of the variable $u$ in $\ell$. We assume that 
$S$ does not divide $f_A$, i.e., $f_A(a) \neq 0$.
The first ingredient we need is the construction of an auxiliary matrix 
which is similar to $A_p(a)$. This construction comes from our previous 
paper~\cite{BoCaSc15}. Let us recall it briefly. We define the ring 
$\Sdp$ of Hurwitz series whose elements are formal infinite sums of the shape:
\begin{equation}
\label{eq:eltSdp}
a_0 + a_1 \gamma_1(t) + a_2 \gamma_2(t) + \cdots + a_n \gamma_n(t)
+ \cdots
\end{equation}
and on which the addition is straightforward and the multiplication
is governed by the rule $\gamma_i(t) \cdot \gamma_j(t) = \binom{i+j}i
\gamma_{i+j}(t)$. (The symbol $\gamma_i(t)$ should be thought of as
$\frac{t^i}{i!}$.) 
We moreover endow $\Sdp$ with the derivation defined by $\gamma_i(t)' = 
\gamma_{i-1}(t)$ (with the convention that $\gamma_0(t) = 1$) and the 
projection map $\text{pr} : \Sdp \to \ell$ sending the series given by 
Eq.~\eqref{eq:eltSdp} to its constant coefficient $a_0$. We shall often 
use the alternative notation
$f(0)$ for $\text{pr}(f)$. 
If $f \in \Sdp$ is given by the series~\eqref{eq:eltSdp}, we then 
have $a_n = f^{(n)}(0)$ for all nonnegative integers~$n$.
We have a homomorphism of rings:
\vspace{-1.5ex}
$$
\psi_S : k[x][{\textstyle \frac{1}{f_A}}] \to \Sdp, \quad
f(x) \mapsto 
\sum_{i=0}^{p-1} f^{(i)}(a) \gamma_i(t).$$

\vspace{-1.5ex}
It is easily checked that $\psi_S$ commutes with the derivation. We 
can then consider the differential module over $\Sdp$ obtained from 
$(M,\partial)$ by scalar extension. By definition, it corresponds to the 
differential system $Y' = \psi_S(A) \cdot Y$.

The benefit of working over $\Sdp$ is the existence of an analogue of 
the well-known Cauchy--Lipschitz Theorem~\cite[Proposition~3.4]{BoCaSc15}.
This notably implies the existence of a fundamental matrix of solutions, 
i.e., an $r\times r$ matrix $Y_S$ with entries in $\Sdp$, and satisfying:
\begin{equation}
\label{eq:defYS}
Y'_S = \psi_S(A) \cdot Y_S
\quad \text{and} \quad
Y_S(0) = \text{I}_r
\end{equation}
with $\text{I}_r$ the identity matrix of size $r$. Moreover, as 
explained in more details later, the construction of $Y_S$ is 
effective.

For any integer $n \geq 0$, we let $Y_S^{(n)}$ denote the matrix 
obtained from $Y_S$ by taking the $n$-th derivative entry-wise. The
next proposition is a consequence of \cite[Proposition~4.4]{BoCaSc15}.

\begin{prop} 
\label{prop:similar}
The matrices $A_p(a)$ and $-Y_S^{(p)}(0)$ are similar over $\ell$.
\end{prop}

\noindent
{\bf Fast computation of $Y_S^{(p)}(0)$.}
We recall that $Y_S$ is defined as the solution of the 
system~\eqref{eq:defYS}. Remembering that we have written $A
= \frac 1{f_A} \tilde A$, we obtain the relation:
\begin{equation}
\label{eq:diffYS}
\psi_S(f_A) \cdot Y_S' = \psi_S(\tilde A) \cdot Y_S.
\end{equation}
Write $f_A = \sum_{i=0}^d f_i \cdot (x-a)^i$ and
$\tilde A = \sum_{i=0}^d \tilde A_i \cdot (x-a)^i$
where the $f_i$'s lie in $\ell$ and the $A_i$'s are square matrices
of size $r$ with entries in $\ell$. Remark that $f_0$ does not vanish
because it is equal to $f_A(a)$. Note moreover that the $f_i$'s can be 
computed for a cost of $\softO(d)$ operations in $k$ using 
divide-and-conquer techniques. Given a fixed pair of indices $(i',j')$, 
the same discussion applies to the collection of the $(i',j')$-entries 
of the $A_i$'s. The total cost for computing the decompositions of $f_A$ 
and $\tilde A$ is then $\softO(dr^2)$.
Now, coming back to the definitions, we find that 
$\psi_S(f_A) = 
\sum_{i=0}^d i! \: f_i \cdot \gamma_i(t)$ and 
$\psi_S(\tilde A) = \sum_{i=0}^d i! \: \tilde A_i \cdot \gamma_i(t)$.
Eq.~\eqref{eq:diffYS} yields the recurrence:
\begin{equation}
\label{eq:recYSn}
Y_S^{(n+1)}(0) = \sum_{i=0}^{\min(n,d)} B_i(n) \cdot Y_S^{(n-i)}(0)
\end{equation}
where the $B_i \in M_r(\ell[u])$ are  defined by:
\begin{equation}
\label{eq:defBi}
f_0 B_i = u (u{-}1) \cdots (u{-}i{+}1) \cdot
\big( \tilde A_i - (u{-}i) f_{i+1} \cdot \text{I}_r \big)
\end{equation}
with the convention that $f_{d+1} = 0$. Now setting:
$$Z_n = \left(
\begin{matrix}
Y_S^{(n-d)}(0) \\
Y_S^{(n-d+1)}(0) \\
\vdots \\
Y_S^{(n)}(0) 
\end{matrix}\right), \quad
B = \left(
\begin{matrix}
& \text{I}_r \\
& & \text{I}_r \\
& & & \ddots \\
& & & & \text{I}_r \\
B_d & \cdots & \cdots & \cdots & B_0
\end{matrix} \right)$$
(with the convention $Y_s^{(i)}(0) = 0$ when $i < 0$), 
the recurrence~\eqref{eq:recYSn} becomes $Z_{n+1} = B(n) 
\cdot Z_n$. Hence, we obtain $Z_p = B(p-1) \cdot B(p-2) \cdots 
B(0) \cdot Z_0$ from what we finally get that $Y_S^{(p)}(0)$ is
the $(r \times r)$-matrix located at the bottom right corner of
$B(p-1) \cdot B(p-2) \cdots B(0)$.
The computation of the former matrix factorial can be performed 
efficiently using a variation of the Chudnovskys' 
algorithm~\cite{ChCh88,BoGaSc07}. Combining this with 
Proposition~\ref{prop:similar}, we end up with the following.

\begin{prop}
\label{prop:invfactlocal}
The invariant factors of $A_p(a)$ can be computed in 
$\softO(d^\omega r^\omega \sqrt{dp})$
operations in the field $\ell$.
\end{prop}

\begin{proof}
Note that $B$ is a square matrix of size $(d{+}1)r$. Moreover coming
back to~\eqref{eq:defBi}, we observe that the entries of~$B$ all have
degree at most $d$. By \cite[Theorem~2]{BoClSa05} the matrix factorial
$- B(p-1) \cdot B(p-2) \cdots B(0)$ can then be computed for the cost
of $\softO(d^\omega r^\omega \sqrt{dp})$ operations in $\ell$. By
\cite{Storjohann01}, the invariant factors of its submatrix
$-Y_S^{(p)}(0)$
can be obtained for an extra 
cost of $\softO(r^\omega)$ operations in $\ell$ (which is negligible 
compared to the previous one). Using Proposition~\ref{prop:similar}
these invariant factors are also those of $A_p(a)$, and we are
done.
\end{proof}

\noindent
{\bf Conclusion.}
Proposition~\ref{prop:invfactlocal} yields an acceptable primitive 
\texttt{invariant\_factors\_at}${}_{A_p(x)}$. Plugging it in the
algorithm \texttt{invariant\_factors\_deterministic} and using the
parameters $D = dr$ and $F = 6dr(r-1)$, we end up with an algorithm
that computes the invariant factors of $A_p(x)$ for the cost of one
unique call to \texttt{invariant\_factors\_at}${}_{A_p(x)}$ with an 
input lying in an extension $\ell/k$ of degree $F+1$ (cf. Proposition~\ref{prop:costdeterministic}).
By Proposition~\ref{prop:invfactlocal}, we find that 
the total complexity of the obtained algorithm is
$\softO\big(d^{\hspace{0.2mm}\omega+\frac 3 2} r^{\omega+2} \sqrt{p}\big)$
operations in $\F_q$.
The first part of Theorem~\ref{theo:pcurvmodules} is then established.
The second part is obtained in a similar fashion using 
the algorithm \texttt{invariant\_factors\_montecarlo} together with 
Proposition~\ref{prop:costmontecarlo} for correctness and complexity 
results.

\subsection{The case of differential operators}

The ring of differential operators $k(x)\!\left<\partial\right>$ is the 
ring of usual polynomials over $k(x)$ in the variable $\partial$ except 
that the multiplication is ruled by the relation
$\partial \cdot f = f \cdot \partial + f'$.
We define similarly the ring $k[x]\!\left<\partial\right>$. We say
that $L \in k[x]\!\left<\partial\right>$ has bidegree $(d,r)$
if it has degree $d$ with respect to $x$ and degree $r$ with respect
to $\partial$.

If $L$ is a differential operator in $k(x)\!\left<\partial\right>$, 
one easily checks that the set 
$k(x)\!\left<\partial\right> L$ of left multiples of $L$ is a left 
ideal of $k(x)\!\left<\partial\right>$. The quotient 
$M_L = k(x)\!\left<\partial\right>/k(x)\!\left<\partial\right>L$
is then a vector space over $k(x)$. It is moreover endowed with a map
$\partial : M_L \to M_L$ given by the left multiplication by $\partial$.
This map turns $M_L$ into a differential module.

We shall prove in this section that the complexities announced in 
Theorem~\ref{theo:pcurvmodules} can be improved in the case of 
differential modules coming from differential operators. Below is 
the statement of our precise result.

\begin{theo}
\label{theo:pcurvoperators}
Let $L \in k[x]\!\left<\partial\right>$ be a differential operator
of bidegree $(d,r)$. We assume $p > r$.

\smallskip

\noindent
1. There exists a deterministic algorithm that computes the
invariant factors of the $p$-curvature of $M_L$ within
$$\softO\big((d+r)^{\omega+1} d^{\frac 1 2} r \cdot \sqrt{p}\big)$$
operations in $k = \F_q$.

\smallskip

\noindent
2. Let $\varepsilon \in (0,1)$.
There exists a Monte-Carlo algorithm that computes the
invariant factors of the $p$-curvature of $M_L$ in
$$\softO\big((d+r)^\omega d^{\frac 1 2} \cdot (d - \log \varepsilon) \cdot \sqrt{p}\big)$$
operations in $k = \F_q$.
This algorithm returns a wrong answer with probability at most
$\varepsilon$.
\end{theo}

\noindent
{\bf Better bounds.}
From now on, we fix a differential operator $L \in k(x)\!\left<
\partial\right>$ of bidegree $(d,r)$. We denote by $A_p(x)$ the
matrix of the $p$-curvature of $M_L$ with respect to the canonical
basis $(1, \partial, \ldots, \partial^{r-1})$. If $a_r(x)$ is the
leading coefficient of $L$ (with respect to~$\partial$),
it follows from~\cite[Proposition~3.2]{Cluzeau03} 
that $A_p(x)$ has the form 
$A_p(x) = \frac 1{a_r(x)^p} \cdot \tilde A_p(x)$
where $\tilde A_p(x)$ is a matrix with polynomial entries of degree
at most $pd$. 

\begin{prop}
\label{prop:Hdoperators}
The matrix $\tilde A_p(x)$ satisfies the hypothesis $(\textbf{H}_{r+d})$
(introduced just before Proposition~\ref{prop:boundbad2}).
\end{prop}

\begin{proof}
This is a direct consequence of Lemma~3.9 and Theorem~3.11 of
\cite{BoCaSc14}.
\end{proof}

\noindent
{\bf The similarity class of $A_p(a)$.}
We now revisit Proposition~\ref{prop:invfactlocal} when 
the differential module comes from the differential operator $L$. We fix 
an irreducible polynomial $S \in k[x]$ and assume that $S$ is coprime
with the leading coefficient $a_r(x)$ of $L$. We set $\ell = k[x]/S$ and 
let $a$ denote the image of $x$ is $\ell$.
We define $t = x - a \in \ell[x]$ and consider the ring of differential
operators $\ell[x]\!\left<\partial\right>$. The latter acts on $\Sdp$
by letting $\partial$ act as the derivation.
Let $Y_S$ be the fundamental system of solutions of the equation
$Y_S' = \psi_S(A) \cdot Y_S$ where $A$ is the companion matrix which
gives the action of $\partial$ on~$M_L$. It takes the form:
$$Y_S = \left(
\begin{matrix}
y_0 & y_1 & \cdots & y_{r-1} \\
y'_0 & y'_1 & \cdots & y'_{r-1} \\
\vdots & \vdots & & \vdots \\
y_0^{(r-1)} & y_1^{(r-1)} & \cdots & y_{r-1}^{(r-1)} 
\end{matrix}\right)$$
where $y_j \in \Sdp$ is the unique solution of the differential equation 
$L y_j = 0$ with initial conditions $y_j^{(n)}(0) = \delta_{j,n}$ (where
$\delta_{\cdot,\cdot}$ is the Kronecker symbol) for $0 \leq n < r$. 

We introduce the Euler operator $\theta = t \cdot \partial \in 
\ell[x]\!\left<\partial\right>$. 
Using the techniques of~\cite[Section~4.1]{BoCaSc14}, one can write 
$L \cdot \partial^d  = \sum_{i=0}^{d+r} b_i(\theta) \partial^i$
within $\softO((r+d)d)$ operations in $\ell$. Here the $b_i$'s are
polynomials with coefficients in $\ell$ of degree at most $d$. One
can check moreover that the polynomial $b_{d+r}$ is constant equal
to $a_r(a)$; in particular, it does not vanish thanks to our
assumption on $S$. 
For all $j$, define $z_j = \sum_{n=0}^\infty y_j^{(n)}(0)
\gamma_{n+d}(t)$. Clearly
$\partial^d z_j = y_j$, so that we have
$\left(\sum_{i=0}^{d+r} b_i(\theta) \partial^i\right) \cdot z_i = 0$
for all $i$. Noting that $\theta$ acts on $\gamma_n(t)$ by multiplication
by $n$, we get the recurrence relation:
\vspace{-1ex}
$$\forall n \geq 0, \quad 
\sum_{i=0}^{d+r} b_i(n) \cdot y_j^{(n+i-d)}(0) = 0$$
with the convention that $y_j^{(n)} = 0$ when $n < 0$. Letting:
\begin{align*}
Z_n & = \left(
\begin{matrix}
y_0^{(n-d)}(0) & \cdots & y_{r-1}^{(n-d)}(0) \smallskip \\
y_0^{(n-d+1)}(0) & \cdots & y_{r-1}^{(n-d+1)}(0) \\
\vdots & & \vdots \\
Y_0^{(n+r-1)}(0) & \cdots & y_{r-1}^{(n+r-1)}(0)
\end{matrix}\right) \in {M}_{d+r,r}(\ell) \medskip \\
\text{and} \quad B & = 
\frac {-1} {a_r(a)} \cdot \left(
\begin{matrix}
& 1 \\
& & \ddots \\
& & & 1 \\
b_0 & b_1 & \cdots & b_{d+r-1}
\end{matrix} \right) \in {M}_{d+r,d+r}(\ell)
\end{align*}
the above recurrence rewrites $Z_{n+1} = B(n) Z_n$. Solving the
recurrence, we get $Z_p = B(p-1) \cdots B(0) \cdot Z_0$, and
we derive that $Y_S^{(p)}(0)$ is the $(r\times r)$ matrix located
at the bottom right corner of $B(p-1) \cdot B(p-2) \cdots B(0)$.
Using Proposition~\ref{prop:similar} and \cite[Theorem~2]{BoClSa05},
we end up with the following proposition (compare with
Proposition~\ref{prop:invfactlocal}).

\begin{prop}
\label{prop:invfactlocal2}
The invariant factors of $A_p(a)$ can be computed in 
$\softO((d+r)^\omega \sqrt{dp})$ operations in the field $\ell$.
\end{prop}

\noindent
{\bf Conclusion.}
The final discussion is now similar to the one we had in the case
of differential modules.
Proposition~\ref{prop:invfactlocal2} provides the primitive 
\texttt{invariant\_factors\_at}${}_{A_p(x)}$. Using it in the
algorithms \texttt{invariant\_factors\_deterministic} and 
\texttt{invariant\_factors\_montecarlo} with the parameters $D = d$ 
and $F = 3d(2r-1)$ 
(coming from the combination of Propositions~\ref{prop:boundbad2} 
and~\ref{prop:Hdoperators}), we respectively end up with
deterministic and Monte-Carlo algorithms whose complexities agree
with the ones announced in Theorem~\ref{theo:pcurvoperators}.

It is instructive to compare the methods and results of this 
section with those of our previous paper~\cite{BoCaSc14}. We
remark that the matrix factorial considered above
is nothing but the specialization at $\theta = 0$ of the matrix
factorial in~\cite{BoCaSc14}.
Although the theoretical approaches of the two papers are definitively 
different, they lead to very similar computations. 
However, each of them has its own advantages and disadvantages.
On the one hand, the methods of~\cite{BoCaSc14} deal 
with characteristic polynomials only and cannot see invariant factors. On the other hand, they do not require the assumption $a_r(a) \neq 0$
(that is why we always took $a = 0$ in~\cite{BoCaSc14})
and can handle at the same time the local computations at the point
$a$ and \emph{around} it, i.e., they provide roughly speaking
a framework which allows to work modulo $(x{-}a)^{pn}$ for some 
integer $n$ fixed in advance (not just $n=1$) without increasing
the cost with respect to $p$.
The practical consequence is that the methods of the current
paper end up with algorithms whose cost is weakened by a 
factor $\sqrt d$ compared to what we might have expected at first.
It would be interesting to find a general theoretical
setting unifying the two approaches discussed above and allowing the 
benefits of both of them.

\section{Solving a physical application}\label{sec:application}

In \cite{BoHaJeMaZe10}, a globally nilpotent differential operator
$\phi_H^{(6)}$ was introduced in order to model the $6$-particle
contribution to the square lattice Ising model. As shown in \emph{loc.\
cit.}, this operator factors as a product of differential operators of
smaller orders. The factor which is the least understood is called
$L_{23}$ and has order $23$. Actually $L_{23}$ has not been computed so
far because its size is too large. Nevertheless there exists a multiple
of $L_{23}$ which has a more reasonable size: its bidegree is $(140,77)$.
It turns out that this multiple, say $L_{77}$, has been determined modulo
several prime numbers.
Based on this computation and using the strategy developed in this 
paper, we were able to study a bit further the factorization of 
$L_{23}$, answering a question of the authors of~\cite{BoHaJeMaZe10}.

\begin{prop}
\label{prop:L23}
The operator $L_{23}$ cannot be factorized as a product $L_{21} 
\cdot L_2$ where
$L_2$ is an operator of order $2$, and
$L_{21}$ is an operator of order $21$ whose differential
module is isomorphic to a symmetric product $\text{Sym}^n M$ for an
integer $n > 1$ and a differential module $M$.
\end{prop}

\begin{proof}
We argue by contradiction by assuming that such a factorization exists.
This would imply that, for all $p$ the matrix $A_{23,p}$ of the $p$-curvature
of $L_{23} \text{ mod } p$ decomposes:
\begin{equation}
\label{eq:decAp}
A_{23,p} = \left(
\begin{matrix} A_{2,p} & \star \\ 0 & A_{21,p} \end{matrix}
\right)
\end{equation}
where $A_{2,p}$ (resp. $A_{21,p}$) is the square matrix of size $2$
(resp.~$21$) and $A_{21,p}$ is similar to a symmetric product of
a $d \times d$ matrix $A_{d,p}$. We now pick $p = 32647$ for which
$L_{77} \text{ mod } p$ is known. 
Using Proposition~\ref{prop:invfactlocal}, we were able to determine 
the invariant factors of the $p$-curvature of $A_{77,p}(15)$.
The computation ran actually rather fast: just a few minutes. 
We observed that the generalized eigenspace of $A_{77,p}(15)$ for
the eigenvalue $0$ has dimension $23$. Combining this with that fact
that $L_{23}$ is a factor of $L_{77}$ whose $p$-curvature is nilpotent, 
we deduce that the restriction of $A_{77,p}(15)$ to this characteristic 
space is similar to $A_{23,p}(15)$. Arguing similarly, we determine 
the Jordan type of $A_{23,p}(15)$:

\smallskip

\renewcommand{\arraystretch}{1.2}

\noindent \hfill
\begin{tabular}{@{}l@{}|c||p{1.5em}|p{1.5em}|p{1.5em}|p{1.5em}|p{1.5em}|p{1.5em}|@{}l@{}}
\cline{2-8} &
\centering $m$ &
\centering $0$ &
\centering $1$ &
\centering $2$ &
\centering $3$ &
\centering $4$ &
{\centering ${\geq}5$} & \\
\cline{2-8} &
rank$\big(A_{23,p}(15)^m\big)$ &
\centering $23$ &
\centering $17$ &
\centering $11$ &
\centering $6$ &
\centering $3$ &
\centering $0$ & \\
\cline{2-8}
\end{tabular}
\hfill \null

\renewcommand{\arraystretch}{1}

\medskip

\noindent
Moreover the writing \eqref{eq:decAp} would imply that for all $m$:
$$\begin{array}{rl}
& 0 \leq
\text{rank}\big(A_{23,p}(15)^m\big) - \text{rank}\big(A_{21,p}(15)^m\big) 
\leq 2 \smallskip \\
\text{and} & 
\text{rank}\big(A_{21,p}(15)^m\big) = 
\displaystyle \binom{n-1 + \text{rank}\big(A_{d,p}(15)^m\big)} n.
\end{array}$$
There is only one way to satisfy these numerical constraints which
consists in taking $n=2$ and:

\smallskip

\renewcommand{\arraystretch}{1.2}
\noindent
\hfill
\begin{tabular}{@{}l@{}|c||p{1.5em}|p{1.5em}|p{1.5em}|p{1.5em}|p{1.5em}|p{1.5em}|@{}l@{}}
\cline{2-8} &
\centering $m$ &
\centering $0$ &
\centering $1$ &
\centering $2$ &
\centering $3$ &
\centering $4$ &
{\centering ${\geq}5$} & \\
\cline{2-8} &
rank$\big(A_{d,p}(15)^m\big)$ &
\centering $6$ &
\centering $5$ &
\centering $4$ &
\centering $3$ &
\centering $2$ &
\centering $0$ & \\
\cline{2-8}
\end{tabular}
\hfill \null
\renewcommand{\arraystretch}{1}

\medskip

\noindent
Since the sequence $ 
\text{rank}\big(A_{d,p}(15)^m\big) - \text{rank}\big(A_{d,p}(15)^{m+1}\big)$ has to be 
non-increasing, this is impossible.
\end{proof}

\vspace{-3mm}

\bibliographystyle{abbrv}

\begin{thebibliography}{10}

\bibitem{Andre04}
Y.~Andr{\'e}.
\newblock Sur la conjecture des {$p$}-courbures de {G}rothendieck-{K}atz et un
  probl\`eme de {D}work.
\newblock In {\em Geometric aspects of {D}work theory. {V}ol. {I}, {II}}, pages
  55--112. Walter de Gruyter GmbH \& Co. KG, Berlin, 2004.

\bibitem{BBHMWZ08}
A.~Bostan, S.~Boukraa, S.~Hassani, J.-M. Maillard, J.-A. Weil, and N.~Zenine.
\newblock Globally nilpotent differential operators and the square {I}sing
  model.
\newblock {\em J. Phys. A}, 42(12):125206, 50, 2009.

\bibitem{BoCaSc14}
A.~Bostan, X.~Caruso, and E.~Schost.
\newblock A fast algorithm for computing the characteristic polynomial of the
  $p$-curvature.
\newblock In {\em I{SSAC}'14}, pages 59--66. ACM, New York, 2014.

\bibitem{BoCaSc15}
A.~Bostan, X.~Caruso, and E.~Schost.
\newblock A fast algorithm for computing the $p$-curvature.
\newblock In {\em I{SSAC}'15}, pages 69--76. ACM, New York, 2015.

\bibitem{BoClSa05}
A.~Bostan, T.~Cluzeau, and B.~Salvy.
\newblock Fast algorithms for polynomial solutions of linear differential
  equations.
\newblock In {\em ISSAC'05}, pages 45--52. ACM Press, 2005.

\bibitem{BoGaSc07}
A.~Bostan, P.~Gaudry, and {\'E}.~Schost.
\newblock Linear recurrences with polynomial coefficients and application to
  integer factorization and {C}artier-{M}anin operator.
\newblock {\em SIAM Journal on Computing}, 36(6):1777--1806, 2007.

\bibitem{BoKa08b}
A.~Bostan and M.~Kauers.
\newblock Automatic classification of restricted lattice walks.
\newblock In {\em {FPSAC}'09}, DMTCS Proc., AK, pages 201--215. 2009.

\bibitem{BoKa08a}
A.~Bostan and M.~Kauers.
\newblock The complete generating function for {G}essel walks is algebraic.
\newblock {\em Proc. Amer. Math. Soc.}, 138(9):3063--3078, 2010.
\newblock With an appendix by Mark van Hoeij.

\bibitem{BoSc09}
A.~Bostan and {\'E}.~Schost.
\newblock Fast algorithms for differential equations in positive
  characteristic.
\newblock In {\em I{SSAC}'09}, pages 47--54. ACM, New York, 2009.

\bibitem{BoHaJeMaZe10}
S.~Boukraa, S.~Hassani, I.~Jensen, J.-M. Maillard, and N.~Zenine.
\newblock High-order {F}uchsian equations for the square lattice {I}sing model:
  {$\chi^{(6)}$}.
\newblock {\em J. Phys. A}, 43(11):115201, 22, 2010.

\bibitem{CaKa91}
D.~G. Cantor and E.~Kaltofen.
\newblock On fast multiplication of polynomials over arbitrary algebras.
\newblock {\em Acta Inform.}, 28(7):693--701, 1991.

\bibitem{ChCh88}
D.~V. Chudnovsky and G.~V. Chudnovsky.
\newblock Approximations and complex multiplication according to {R}amanujan.
\newblock In {\em Ramanujan revisited (Urbana-Champaign, 1987)}, pages
  375--472. Academic Press, Boston, 1988.

\bibitem{Cluzeau03}
T.~Cluzeau.
\newblock Factorization of differential systems in characteristic {$p$}.
\newblock In {\em ISSAC'03}, pages 58--65. ACM Press, 2003.

\bibitem{Eberly00}
W.~Eberly.
\newblock Asymptotically efficient algorithms for the {F}robenius form.
\newblock Technical report, 2000.

\bibitem{LeGall14}
F.~L. Gall.
\newblock Powers of tensors and fast matrix multiplication.
\newblock In {\em ISSAC'14}, pages 296--303, 2014.

\bibitem{Gantmacher59}
F.~R. Gantmacher.
\newblock {\em The theory of matrices. {V}ols. 1, 2}.
\newblock Translated by K. A. Hirsch. Chelsea Publishing Co., New York, 1959.

\bibitem{GaGe03}
J.~\gathen{von zur} Gathen and J.~Gerhard.
\newblock {\em Modern Computer Algebra}.
\newblock Cambridge University Press, second edition, 2003.

\bibitem{GiesbrechtPhD}
M.~Giesbrecht.
\newblock {\em Nearly optimal algorithms for canonical matrix forms}.
\newblock PhD thesis, University of Toronto, 1993.

\bibitem{Giesbrecht95}
M.~Giesbrecht.
\newblock Nearly optimal algorithms for canonical matrix forms.
\newblock {\em SIAM Journal on Computing}, 24(5):948--969, 10 1995.

\bibitem{GiSt02}
M.~Giesbrecht and A.~Storjohann.
\newblock Computing rational forms of integer matrices.
\newblock {\em Journal of Symbolic Computation}, 34(3):157--172, 2002.

\bibitem{HaHoLe14}
D.~Harvey, J.~van~der Hoeven, and G.~Lecerf.
\newblock Faster polynomial multiplication over finite fields.
\newblock \url{http://arxiv.org/abs/1407.3361}, 2014.

\bibitem{Honda81}
T.~Honda.
\newblock Algebraic differential equations.
\newblock In {\em Symposia {M}athematica, {V}ol. {XXIV}}, pages 169--204.
  Academic Press, London, 1981.

\bibitem{KaKrSa87}
E.~Kaltofen, M.~S. Krishnamoorthy, and B.~D. Saunders.
\newblock Fast parallel computation of {H}ermite and {S}mith forms of
  polynomial matrices.
\newblock {\em SIAM Journal on Matrix Analysis and Applications},
  8(4):683--690, 1987.

\bibitem{KaKrSa90}
E.~Kaltofen, M.~S. Krishnamoorthy, and B.~D. Saunders.
\newblock Parallel algorithms for matrix normal forms.
\newblock {\em Linear Algebra and its Applications}, 136:189--208, 1990.

\bibitem{KaVi04}
E.~Kaltofen and G.~Villard.
\newblock On the complexity of computing determinants.
\newblock {\em Comput. Complexity}, 13(3-4):91--130, 2004.

\bibitem{Katz72}
N.~M. Katz.
\newblock Algebraic solutions of differential equations ({$p$}-curvature and
  the {H}odge filtration).
\newblock {\em Invent. Math.}, 18:1--118, 1972.

\bibitem{Katz82}
N.~M. Katz.
\newblock A conjecture in the arithmetic theory of differential equations.
\newblock {\em Bull. Soc. Math. France}, (110):203--239, 1982.

\bibitem{PeSt07}
C.~Pernet and A.~Storjohann.
\newblock Frobenius form in expected matrix multiplication time over
  sufficiently large fields, 2007.

\bibitem{Schoenhage77}
A.~Sch{\"o}nhage.
\newblock Schnelle {M}ultiplikation von {P}olynomen \"uber {K}\"orpern der
  {C}harakteristik 2.
\newblock {\em Acta Informatica}, 7:395--398, 1977.

\bibitem{Storjohann01}
A.~Storjohann.
\newblock Deterministic computation of the {F}robenius form.
\newblock In {\em FOCS'01}, pages 368--377. IEEE Computer Society Press, 2001.

\bibitem{StLa97}
A.~Storjohann and G.~Labahn.
\newblock A fast {L}as {V}egas algorithm for computing the {S}mith normal form
  of a polynomial matrix.
\newblock {\em Linear Algebra and its Applications}, 253(1–3):155--173, 1997.

\bibitem{Put95}
M.~{van der Put}.
\newblock Differential equations in characteristic $p$.
\newblock {\em Compositio Mathematica}, 97:227--251, 1995.

\bibitem{Put96}
M.~{van der Put}.
\newblock Reduction modulo $p$ of differential equations.
\newblock {\em Indag. Mathem.}, 7(3):367--387, 1996.

\end{thebibliography}

\end{document}